\documentclass[sigconf,natbib=false]{acmart}

\usepackage{amsmath,amsfonts}
\usepackage{array}      %
\usepackage{comment}  %
\usepackage{xcolor}
\usepackage{hyperref}
\usepackage{float}
\usepackage{url}
\usepackage{hyperref}

\usepackage[utf8]{inputenc}

\usepackage{caption}      %

\floatstyle{ruled}
\newfloat{algo}{tp}{lop}
\floatname{algo}{Algorithm}
\usepackage[noend]{algpseudocode}
\algrenewcommand\textproc{\sffamily}
\algrenewcommand\algorithmicindent{1em}

\usepackage[nocompress,nosort]{cite}

\usepackage[shortlabels,inline]{enumitem} %

\usepackage{flushend}   %
\usepackage{mathtools}  %
\usepackage{booktabs}   %

\makeatletter
\def\@acmplainindent{0pt}
\def\@acmdefinitionindent{0pt}
\def\@proofindent{\noindent}
\makeatother

\theoremstyle{acmdefinition}

\makeatletter
\renewcommand\subsubsection{\@startsection{subsubsection}{3}{10pt}%
  {\smallskipamount}%
  {-3.5\p@}%
  {\@subsubsecfont\@adddotafter}}
\renewcommand\paragraph{\@startsection{paragraph}{4}{\parindent}%
  {\smallskipamount}%
  {-3.5\p@}%
  {\@parfont\@adddotafter}}
\makeatother

\def\KK{\mathbb{K}}
\def\QQ{\mathbb{Q}}
\def\calg{\overline{\mathclap{\phantom{t}}\smash{\KK}}}

\def\ind{\operatorname{ind}}

\def\ud{\mathrm{d}}
\def\Exc{\operatorname{Exc}_M}
\def\im{\operatorname{im}}

\def\st{\ \middle|\ }
\def\ord{\operatorname{ord}}
\def\Sing{\operatorname{Sing}}
\def\ann{\operatorname{ann}}

\def\wrt{w.r.t. }
\def\eqsmall{\fontsize{8pt}{\baselineskip}\selectfont}

\title[Generalized Hermite Reduction, Creative Telescoping and Definite Integration]{Generalized Hermite Reduction, Creative Telescoping \\ and Definite Integration of D-Finite Functions}

\author{Alin Bostan}
 \affiliation{%
   \institution{Inria, France}
 }
\email{alin.bostan@inria.fr}
\author{Frédéric Chyzak}
 \affiliation{%
   \institution{Inria, France}
 }
\email{frederic.chyzak@inria.fr}
\author{Pierre Lairez}
 \affiliation{%
   \institution{Inria, France}
 }
\email{pierre.lairez@inria.fr}

\author{Bruno Salvy}
\affiliation{%
  \institution{Inria, France}
}
\additionalaffiliation{\institution{Univ Lyon, Inria, CNRS, ENS de Lyon, Université Claude Bernard Lyon~1, LIP UMR 5668, F-69007 Lyon, France}}
\email{bruno.salvy@inria.fr}

\begin{document}

\copyrightyear{2018}
\acmYear{2018}
\setcopyright{othergov}
\acmConference[ISSAC '18]{2018 ACM International Symposium on Symbolic and Algebraic Computation}{July 16--19, 2018}{New York, NY, USA}
\acmBooktitle{ISSAC '18: 2018 ACM International Symposium on Symbolic and Algebraic Computation, July 16--19, 2018, New York, NY, USA}
\acmPrice{15.00}
\acmDOI{10.1145/3208976.3208992}
\acmISBN{978-1-4503-5550-6/18/07}

\begin{abstract} Hermite reduction is a classical algorithmic tool in symbolic
integration. It is used to decompose a given rational function as a sum of a
function with simple poles and the derivative of another rational function.
We extend Hermite reduction to arbitrary linear differential operators
instead of the pure derivative,
and develop efficient algorithms for this reduction. We
then apply the generalized Hermite reduction to the computation of linear
operators satisfied by single definite integrals
of D-finite functions of several continuous or discrete parameters.
The resulting algorithm is a generalization of reduction-based methods
for creative telescoping.
\end{abstract}

\maketitle

\section{Introduction}

Ostrogradsky\footnote{Most references in symbolic
integration attribute to Ostrogradsky an algorithm to compute $U$ and $A$
based on linear algebra. As a matter of fact, Ostrogradsky introduced before
Hermite a polynomial method, based on extended gcds. In passing, he invented
an efficient algorithm for squarefree factorization, rediscovered by
Yun~\cite{Yun76,Yun77}.}~\cite{Ostrogradsky1845} and Hermite~\cite{Hermite1872} showed how to
decompose the indefinite integral $\int R $ of a rational function~$R \in
\QQ(x)$ as~$U + \int A$, where~$U, A\in\QQ(x)$, and where~$A$ has only simple
poles and vanishes at infinity. Their contributions consist in \emph{rational
algorithms} to compute~$A$ and~$U$, that is algorithms which do not require to
manipulate the roots in $\smash{\overline{\QQ}}$ of the denominator of~$R$, but merely
its (squarefree) factorization. The rational function $A$ is classically
called the \emph{Hermite reduction} of~$R$.
In other words, the Hermite reduction of~$R$ is a \emph{canonical form} of~$R$
modulo the derivatives in~$\QQ(x)$: it depends $\QQ$-linearly on~$R$, it is
equal to~$R$ modulo the derivatives and it vanishes if and only if~$U' = R$
for some~$U\in\QQ(x)$.

We call \emph{generalized Hermite problem} the analogous question for
inhomogeneous linear differential equations of arbitrary order
\begin{equation}\label{eq:inhomdeq}
	c_r(x)y^{(r)}(x)+\dots+c_0(x)y(x)=R(x),
\end{equation}
where~$R$ and the~$c_i$ are rational functions in~$\KK(x)$, over some
field~$\KK$ of characteristic zero. In operator notation,
given~$L=c_r\partial_x^r+\dots+c_0\in\KK(x)\langle\partial_x\rangle$, the
problem is to produce a rational function~$[R]$ in~$\KK(x)$, that depends
$\KK$-linearly on~$R$, that is equal to~$R$ modulo the image~$L(\KK(x))$ and
that vanishes if and only if~$R$ is in~$L(\KK(x))$.

Equations like Eq.~\eqref{eq:inhomdeq} occur in relation to \emph{integrating
factors}, and ultimately to \emph{creative telescoping}. If $L^*$
denotes the \emph{adjoint} of~$L$, defined
as $L^* = \sum_{i=0}^r (-\partial_x)^i c_i(x)$, then for any
function~$f$, integration by parts produces Lagrange's
identity~\cite[\S5.3]{Ince44}
\begin{equation}\label{eq:Lagrange}
uL(f) - L^*(u) f = \partial_x \left( P_L(f,u) \right),
\end{equation}
where $P_L$ depends linearly on
$f,\ldots,f^{(r-1)},u,\ldots,u^{(r-1)}$. It follows that if $f$ is a
solution of $L$, then any~$R \in L^*(\KK(x))$ is an integrating factor of~$f$,
meaning that $Rf$ is a derivative of a $\KK(x)$-linear combination of~$f$
and its derivatives.
The converse holds if~$L$ is an operator of minimal order
canceling~$f$, see Proposition~\ref{prop:novector}.

\subsection*{Contributions}
We introduce a \emph{generalized Hermite reduction} to compute such
a~$[R]$. Classical Hermite reduction addresses the case~$L=\partial_x$. The
algorithm operates locally at each singularity and it avoids algebraic
extensions, similarly to classical Hermite reduction.

Next, we improve Chyzak's algorithm~\cite{Chyzak00}
for creative telescoping with the use of generalized Hermite reduction.
Recall that creative telescoping is an
algorithmic way to compute integrals by repeated differentiation under the
integral sign and integration by parts~\cite{AlZe90}.
Chyzak's  algorithm repeatedly checks for the existence of a
rational solution to equations like~\eqref{eq:inhomdeq}.
A lot of time is spent checking that none exists.
The use of generalized Hermite reduction makes the computation
incremental and less redundant.

As a simple instance of the creative telescoping problem, let~$f(t,x)$ be a
function annihilated by a linear
differential operator $L\in\QQ(t,x)\langle\partial_x\rangle$ in the
differentiation with respect to $x$
only, and such that $\partial_t(f)=A(f)$ for another operator~$A$ also in
$\QQ(t,x)\langle\partial_x\rangle$.
We look for the minimal relation of the form
\begin{equation}\label{eq:5}
\lambda_0f+\dots+\lambda_s\partial_t^s(f)=\partial_x(G),
\end{equation}
with $\lambda_0,\dotsc,\lambda_s\in \QQ(t)$ and~$G(t,x)$ in the function space
spanned by~$f$ and its derivatives,
with the motive that integrating both sides with respect to~$x$ may lead to
something useful:
on the right-hand side, the integral of the derivative simplifies, often to~$0$,
and on the left-hand side, the integration commutes with the~$\lambda_i
\partial_t^i$, yielding a differential equation for $\smash{\int} f(t,x)\,\ud
x$.
In Equation~\eqref{eq:5}, the left-hand side is called the \emph{telescoper}
and the function~$G(t,x)$ the \emph{certificate}.

The new algorithm constructs a sequence of rational functions
$R_0, R_1,\dotsc$ in $\QQ(t,x)$ such that $\partial_t^i(f) =
R_if+\partial_x(\dotsc)$. Equation~\eqref{eq:5} holds
if and only if~$\lambda_0 R_0 +\dotsb + \lambda_s R_s$ is an integrating factor
of~$f$,
which in turn is equivalent to the relation
\begin{equation}\label{eq:6}
  \lambda_0[R_0]+\dots+\lambda_s[R_s]=0,
\end{equation}
where~$[\ ]$ is the generalized Hermite reduction with respect to~$L^*$.
Starting with~$s=0$, we search for solutions of the equation above and
increment~$s$ until one is found.
Chyzak's algorithm would solve Equation~\eqref{eq:5} at each iteration mostly from
scratch, whereas the new algorithm
retains the reduced forms $[R_i]$ from one iteration to the next,
computes~$[R_{s}]$ from~$[R_{s-1}]$ and solves the straightforward
Equation~\eqref{eq:6}.
This approach to creative telescoping generalizes to several
parameters~$t_1,\dotsc,t_e$ in the integrand and to different kinds of operators acting on them, in the setting of
Ore algebras.

The order of the telescopers and even the termination of the creative telescoping process are related to the confinement properties of the generalized Hermite
reduction. Assuming that the poles of the rational functions~$R_0,R_1,\dotsc$ all lie in the same finite set, we deduce from a result of Adolphson's an upper bound on the dimension of the subspace spanned by the reductions~$[R_i]$, which in turn bounds the order of the minimal telescoper.

\subsection*{Previous work}

\paragraph{Extensions of Hermite reduction.}
Ostrogradsky~\cite{Ostrogradsky1845} and Hermite~\cite{Hermite1872} introduced
a reduction for rational functions. A century later, it was extended to larger
and larger classes of functions: algebraic~\cite{Trager84},
hypergeometric~\cite{AbPe01},
hyperexponential~\cite{Davenport86,GeLeLi04,BoChChLiXi13,BoDuSa16},
Fuchsian~\cite{ChHoKaKo18}. Van der Hoeven's preprint~\cite{Hoeven17}
considers a reduction w.r.t. the derivation operator on
differential modules of finite type, so as to address the general
differentially
finite case. Our
generalized Hermite reduction is inspired by these works. It has the same
architecture as several previous ones\cite{BoChChLiXi13,ChHoKaKo18,BoDuSa16,Hoeven17}: local reductions at finite places, followed
by a reduction at infinity and the computation of an exceptional set to
obtain a canonical form.
Our first contribution in the present paper is to open
a new direction of generalization, namely by considering reductions
with respect to other operators in $\KK(x)\langle\partial_x\rangle$ than
the derivation operator~$\partial_x$,
acting on the space $\KK(x)$ of rational
functions.
An extra benefit of our method is to avoid algebraic extensions of~$\KK$.

\paragraph{Index theorems.}
The finite-dimensionality of a function space
modulo the image of a differential operator
is crucial to the termination of our reduction and creative-telescoping algorithms.
This finiteness, and even explicit bounds, are given by \emph{index theorems}
for differential equations~\cite{Malgrange74}.
Rational versions
appeared in work by Monsky~\cite{Monsky72} related to the finiteness of de
Rham cohomology, and by Adolphson~\cite{Adolphson76} in a $p$-adic context,
see also~\cite{Rezaoui01,PuRe08}, and~\S\ref{sec:bound-dimens-quot}.

\paragraph{Creative telescoping by reduction.} The use of Hermite-like
reductions for computing definite integrals roots in works by
Fuchs~\cite{Fuchs1870} and Picard~\cite{Picard1902,PiSi1897}.
In the realm of creative telescoping,
this line of research forms what is called
the fourth generation of creative telescoping algorithms.
It was first introduced for bivariate rational functions~\cite{BoChChLi10},
and later extended to the multivariate rational case~\cite{BoLaSa13,Lairez16}.
For bivariate functions/sequences, the approach was also extended
to larger classes:
algebraic~\cite{ChKaSi12,ChKaKo16},
hyperexponential~\cite{BoChChLiXi13},
hypergeometric~\cite{ChHuKaLi15,Huang16},
mixed~\cite{BoDuSa16},
Fuchsian~\cite{ChHoKaKo18},
differentially finite~\cite{Hoeven17}.
Our second contribution is the first reduction-based variant,
for single integrals,
of Chyzak's algorithm~\cite{Chyzak00}
for D-finite functions depending
on several continuous or discrete parameters.

\section{Introductory Example}
\label{sec:an-intr-deta}

\subsection{Hermite Reduction}\label{sec:ex-hermite-reduction}
The equation~$M(y)= a x^2+ bx + c $, with $M$ defined by
\[M(y)=(x^2-1)y''+(x-2p(x^2-1))y'+(p^2(x^2-1)-px-n^2)y,\]
has a rational solution~$y \in \QQ(n,p,x)$ if and only if
$a x^2 + bx + c$ is a multiple of~$p^2 x^2 - px - n^2-p^2$. %
This follows in two steps.

First,
a local analysis reveals that if~$y$ has a pole at some~$\alpha\in\mathbb{C}$,
then so does~$M(y)$:
for any~$\alpha\in\mathbb{C}\setminus \{\pm 1\}$ and for any~$s > 0$,
\begin{align*}
  M\!\left( (x-\alpha)^{-s} \right) &= (\alpha^2-1) s (s+1)  (x-\alpha)^{-s-2} \left(1 + O(x-\alpha) \right) \\
  \text{and }  M\!\left( (x\pm 1)^{-s} \right) &= \pm s(2s+1)(x\pm1)^{-s-1}\left(1+O(x\pm 1)\right).
\end{align*}
Therefore, if~$M(y)$ is a polynomial then~$y$ is also a polynomial.

Next, for any~$s \geq 0$,
$ M(x^s) = p^2 x^{s+2}+O(x^{s+1})$, as $x\to \infty$.
It follows that if~$y \in \QQ(n,p)[x]$ then~$M(y) \in \QQ(n,p)[x]$ and~$\deg_x
M(y) = \deg_x y + 2$. In particular, every polynomial of degree $\leq$ 2 in~$M(\QQ(n,p,x))$ is
a multiple of~$M(1)=p^2 x^2 - px - n^2-p^2$ over~$\QQ(n,p)$.

In~\S \ref{sec:canonical-forms}, we define the Hermite reductions \wrt $M$ of~$1$,
$x$ and~$x^2$:
\[ [1] = 1, \quad [x] = x, \quad \text{and} \quad [x^2] = \frac{x}{p} +
  \frac{n^2+p^2}{p^2}, \]
showing that $[p^2 x^2 - px - n^2-p^2] = 0$.
Similarly, the reduction of any polynomial \wrt~$M$
is a $\QQ(n,p)$-linear combination of $1$ and~$x$.

\subsection{Creative Telescoping}\label{sec:ex-ct}
We consider the classical integral identity
\cite[\S2.18.1, Eq. (10)]{PrBrMa88}
\[\int_{-1}^{1}{\frac{e^{-px}T_n(x)}{\sqrt{1-x^2}}\,\ud x}=(-1)^n\pi
I_n(p),\]
where~$T_n$ denotes the $n$th Chebyshev polynomial of the first kind and $I_n$
the $n$th modified Bessel function of the first kind.
The integrand~$F_n (p,x)$
satisfies a system of linear differential and difference equations, easily found from
defining equations for~$T_n(x)$ and $e^{-px}$:
\begingroup\eqsmall
  \begin{gather*}
  \frac{\partial F_n}{\partial p}=-xF_n, \quad
  nF_{n+1}=\frac{\partial}{\partial x}\left((x^2-1)F_n\right)+(p
  x^2+(n-1)x-p)F_n,\\
 (1-x^2)\frac{\partial^2F_n}{\partial x^2} =  (2px^2+3x-2p)
\frac{\partial F_n}{\partial x}+(p^2x^2+3px-n^2-p^2+1)F_n.
\end{gather*}
\endgroup

We aim at finding a similar set of
linear differential-difference
operators in the variables~$n$ and~$p$ for the integral
$\int_{-1}^1 F_n(p,x)\,\ud x$.
Note that~$F_n$
and all its derivatives \wrt $x$ and~$p$ and shifts
\wrt~$n$ are $\QQ(n,p,x)$-linear combinations of~$F_n$ and~$\partial
F_n/\partial x$.

The \emph{adjoint} of the last equation
is~$M(y)=0$, with the operator~$M$ of \S\ref{sec:ex-hermite-reduction}.
The reduction \wrt $M$ described above makes the following computation possible.
First, $F_n$ is not a derivative (of a $\QQ(n,p,x)$-linear combination of~$F_n$ and~$\partial F_n/\partial x$).
Indeed, $F_n$ is a derivative if and only if~$1 \in M(\QQ(n,p))$.
Second, no $\QQ(n,p)$-linear relation between~$F_n$
and~${\partial F_n}/{\partial p}$ is a derivative, because ${\partial F_n}/{\partial p} = -x
F_n$ and $[1]$ and~$[-x]$ are linearly independent over $\QQ(n,p)$.
Third, the $\QQ(n,p)$-linear relation $p^2 [x^2] + p[-x] - (n^2+p^2)[1] = 0$
proves that
\begin{equation}\label{eq:ex-reldp2}
p^2\frac{\partial^2F_n}{\partial p^2}+p\frac{\partial F_n}
{\partial
p}-(n^2+p^2) F_n =\frac{\partial G}{\partial x}
\end{equation}
for some $\QQ(n,p,x)$-linear combination $G$ of $F_n$ and $\partial
F_n/\partial x$. Next, the equation for~$n F_{n+1}$
and the equation
$[px^2 + (n-1)x - p] = nx + n^2/p$
show that, for some $\tilde G$ as above,
\begin{equation}\label{eq:2}
F_{n+1}+\frac{\partial F_n}{\partial p}-\frac{n}{p}F_n=
\frac{\partial \tilde G}{\partial x}.
\end{equation}

Equations~\eqref{eq:ex-reldp2} and~\eqref{eq:2} can then be integrated from~$-1$
to~1. The contour can be deformed so that the right-hand sides vanish (regardless of~$G$ and~$\tilde G$) and the left-hand sides
provide the desired operators for the integral.
These equations classically define, up to a constant factor, the function
$(-1)^n I_n(p)$.

\section{Generalized Hermite reduction}
\label{sec:canonical-forms}

Throughout this section, $M \in\KK[x] \langle \partial_x \rangle$ denotes a
linear differential operator with polynomial coefficients. We are interested
in finding $\KK$-linear dependency relations in~$\KK(x)$ modulo the rational
image~$M(\KK(x))$ by means of a \emph{canonical form with respect to $M$}.

\begin{definition}\label{def:canonical}
  A \emph{canonical form with respect to $M$} is a  $\KK$-linear map~$[\ ] :
  \KK(x) \to \KK(x)$ such that for any~$R\in \KK(x)$:
  \begin{quote}
    \begin{enumerate*}[(i)]
    \item \ $[M(R)] = 0$;\quad\quad\quad \mbox{}
    \item \ $R - [R] \in M(\KK(x))$.
    \end{enumerate*}
  \end{quote}
\end{definition}
\noindent Applying~$[\ ]$ to~$R - [R]$ before using (ii) and~(i) results in $[[R]] = [R]$.

As can be seen from Eq.~\eqref{eq:inhomdeq},
computing such canonical forms is tightly
related to the computation of rational solutions of linear
differential equations.
In classical solving algorithms~\cite{AbKv91,Liouville33},
bounds on the order of poles of meromorphic solutions are given by
indicial equations. Next, in order to factor the computation
for different inhomogeneous parts, instead of using a ``universal denominator'',
one could at each singularity identify the polar behaviour of
potential meromorphic solutions,
so as to reduce rational solving to polynomial solving.
This idea is what inspired the reduction algorithm for computing canonical forms
in the present section%
\footnote{In the case of systems,
analogues of indicial equations are more complicated;
several alternatives for rational solving exist~\cite{Abramov99,Barkatou99},
that resemble the reduction %
in~\cite{Hoeven17}.}.

\medskip

We begin in \S
\ref{sec:shifts-indic-equat} with a local analysis of~$M(\KK(x))$. Then we describe in \S
\ref{sec:weak-herm-reduct} a projection
map~$H : \KK(x) \to \KK(x)$ that we call \emph{weak Hermite reduction}. It is
not quite a canonical form. It misses an \emph{exceptional set} described
in~\S \ref{sec:canonical-form}, from which a canonical form is deduced. For
simplicity, this is first described in the algebraic closure of the base
field~$\KK$, and in \S \ref{sec:comp-an-herm} we show how to perform the
computations in a rational way, i.e., without algebraic extensions.

Finally, in \S \ref{sec:bound-dimens-quot}, we bound the dimension of the
quotient~$E/M(E)$,
for a ring~$E$ of rational functions with prescribed poles.
This is %
relevant to getting size %
and complexity
bounds for creative telescoping.

\subsection{Local Study}
\label{sec:shifts-indic-equat}

Let~$\calg$ be an algebraic closure of~$\KK$.
For~$R\in\KK(x)$ and~$\alpha\in\calg$, let~$R_{(\alpha)}$ denote the \emph{polar
  part} of~$R$ at~$\alpha$. This is the unique polynomial
in~$\smash{(x-\alpha)^{-1}}$ with constant term zero
such that $R - R_{(\alpha)}$ has no pole at~$\alpha$.
Similarly, the \emph{polynomial part} $R_{(\infty)}$ of~$R$ is the unique
polynomial such that~$R - R_{(\infty)}$ vanishes at infinity.
By partial fraction decomposition,
\begin{equation}\label{eq:partialfraction}
 R = R_{(\infty)} + \sum_{\alpha\in\calg} R_{(\alpha)}.
\end{equation}
Let also~$\ord_\alpha R$ denote the valuation of~$R$ as a Laurent series in~$x-\alpha$.

For any~$\alpha\in\calg$, there exists a non-zero polynomial~$\ind_\alpha \in \calg[s]$
and an integer~$\sigma_\alpha$ such that for any~$s\in\mathbb{Z}$,
\begin{equation}\label{eq:1}
M\!\left( (x-\alpha)^{-s} \right) = \ind_\alpha(-s)(x-\alpha)^{-s+\sigma_\alpha} ( 1 + o(1) ),\quad \text{as $x\to\alpha$.}
\end{equation}
The polynomial $\ind_\alpha$ is classically called the \emph{indicial polynomial
  of~$M$ at $\alpha$} \cite{PuSi03,Ince44};
we call the integer $\sigma_\alpha$ the \emph{shift of~$M$ at $\alpha$}.
The indicial polynomial and its integer roots give a detailed understanding of the
image of $M$. %
We similarly define the shift and the indicial polynomial at~$\infty$ by the
equation
\[ M(x^s) = \ind_\infty(-s) x^{s-\sigma_\infty} (1+o(1)), \quad \text{as $x\to \infty$.} \]
If $M = \sum_{i=0}^r p_i(x) \partial^i_x$, then
\[ \sigma_\alpha = \min_{0\leq i\leq r} (\operatorname{ord}_\alpha p_i
- i)
  \quad\text{and}\quad \sigma_\infty = \max_{0\leq i\leq r}(i-\deg p_i). \]

For any~$\alpha\in\calg$ that is not a root of the leading coefficient $p_r$
of~$M$, we have $\ind_\alpha(s) =  p_r(\alpha) \cdot
s(s-1)\dotsm(s-r+1)$ and~$\sigma_\alpha = -r$.

\subsection{Weak Hermite Reduction}
\label{sec:weak-herm-reduct}

Let~$\im M = M(\KK(x))$.
Let~$H_\alpha : \calg(x) \to \calg(x)$
be the \emph{local reduction map at~$\alpha$} defined
by~$H_\alpha(R) = R$ if~$\ord_\alpha R \geq 0$ ($\alpha$ is not a pole of~$R$)
and by induction on~$\ord_\alpha R$,
\begin{equation*}
  H_\alpha(R) =
  \begin{cases}
    H_\alpha\!\left( R - \frac{cM( (x-\alpha)^{-s-\sigma_\alpha})}{\ind_\alpha(-\sigma_\alpha - s)}  \right) & \hfill \text{\llap{if $\ind_\alpha(-\sigma_\alpha - s) \neq 0$,}}  \\
    c(x-\alpha)^{-s} + H_\alpha\!\left( R - c(x-\alpha)^{-s} \right) & \quad\quad\text{otherwise,}
  \end{cases}
\end{equation*}
where $R = c (x-\alpha)^{-s} \left( 1+ o(1) \right)$ as~$x\to\alpha$,
with~$c\in\calg\setminus \{0\}$
and~$s > 0$.
The induction is well-founded because in either case of the definition,
the argument of~$H_\alpha$ in the right-hand side has a valuation at~$\alpha$ that is larger
than~$\ord_\alpha R$.
By construction, we check that
$R - H_\alpha(R) \in \im M$ for any~$R\in \KK(x)$.

Similarly, let~$H_\infty : \KK(x) \to \KK(x)$ be the \emph{local reduction map at~$\infty$} defined by
$H_\infty(R) = R$ if~$\ord_\infty R > 0$ (that is~$R_{(\infty)} = 0$) and by
induction on~$\ord_\infty(R)$ by
\begin{equation*}
  H_\infty(R) =
  \begin{cases}
    H_\infty\!\left( R- \frac{cM( x^{s+\sigma_\infty} )}{\ind_\infty(-s-\sigma_\infty)} \right) & \text{\parbox{3cm}{\raggedleft if $\ind_\infty(-s-\sigma_\infty) \neq 0$\\
        and $s+\sigma_\infty \geq 0$,}} \\
    c x^s + H_\infty\!\left( R - c x^s \right) & \hfill\text{otherwise,}
  \end{cases}
\end{equation*}
where~$R = cx^s \left(1 + o(1) \right)$ as~$x\to\infty$.
By construction, we check that
$R - H_\infty(R) \in \im M$ for any~$R\in \KK(x)$.
The condition~$s+\sigma_\infty \geq 0$ ensures that~$M(x^{s+\sigma_\infty})$ is a
polynomial.

\begin{definition}
  The \emph{weak Hermite reduction} is the linear map~$H$,
  seen either as~$H :\KK(x)\to\KK(x)$ or as~$H:\calg(x)\to\calg(x)$,
  and defined by
\[ H(R) = H_\infty \bigg( R_{(\infty)} + \sum_{\alpha\in \left\{ \text{poles of~$R$} \right\}} H_\alpha \big( R_{(\alpha)} \big) \bigg). \]
\end{definition}

\begin{proposition}\label{prop:hermite-reduced}
  For any~$R\in \KK(x)$:
  \begin{enumerate}[(i)]
  \item\label{item:H:altdef} $H(R) = H_\infty \circ H_{\alpha_1} \circ \dotsb
    \circ H_{\alpha_n} (R)$, where~$\alpha_1,\dotsc,\alpha_n \in \calg$ are the poles of~$R$;
  \item\label{item:H:equiv-mod-L} $R - H(R) \in \im M$ and $H(M(R)) \in \im M$;
  \item\label{item:H:idempotent} $H(H(R)) = H(R)$.
  \end{enumerate}
  Moreover:
  \begin{enumerate}[(i), resume]
  \item\label{item:H:reduceTa} for any~$\alpha \in \calg$ and for any~$s > 0$,
    \[\ind_\alpha(s) \neq 0  \text{ and }  \sigma_\alpha - s > 0\Rightarrow H\!\left(M((x-\alpha)^{-s})\right) = 0; \]
  \item\label{item:H:reduceinf} for any $s \geq 0$,
    $\displaystyle \ind_\infty(s) \neq 0 \Rightarrow H\!\left(M(x^s)\right) = 0$.
  \end{enumerate}
\end{proposition}

\begin{proof}
  By linearity and Equation~\eqref{eq:partialfraction},
  Property~\ref{item:H:altdef} follows from the formulas~$H(R_{(\infty)}) =
  H_\infty ( R_{(\infty)} )$ and~$H(R_{(\alpha)}) =
  H_\infty(H_\alpha(R_{(\alpha)}))$ derived from the definition of~$H$.
  The first part of Property~\ref{item:H:equiv-mod-L} follows
  from corresponding properties for~$H_\alpha$ and~$H_\infty$;
  the second part is a consequence of applying the first to~$M(R)$.

  As for the idempotence,
  we observe, first, that every~$H(R)$ is a linear combination of
  some~$(x-\alpha)^{-s}$, with~$\ind_\alpha(-s-\sigma_\alpha)=0$,
  and~$x^s$, with~$s + \sigma_\infty \geq 0$
  and~$\ind_\infty(-s-\sigma_\infty)=0$;
  and second, that~$H$ is the identity on such monomials.

  As for~\ref{item:H:reduceTa}, the condition
   $\ind_\alpha(s) \neq 0$ together with \eqref{eq:1} imply
  that~$\ord_\alpha M((x-\alpha)^{-s}) = -s-\sigma_\alpha$,
  and then by definition of~$H_\alpha$,
  \[ H_\alpha\!\left( M((x-\alpha)^{-s} )\right) = H_\alpha\!\left(
      M((x-\alpha)^{-s}) - M((x-\alpha)^{-s} )\right) = 0. \]
  The last property is proved similarly.
\end{proof}

\subsection{Canonical Form}
\label{sec:canonical-form}
If $H$ were a canonical form, $H(M(R))$ would be~0 for any~$R\in\KK(x)$.
But this property fails, and more work is required to refine $H$ into a
canonical form.
\begin{definition}
  The space $\Exc$ of \emph{exceptional functions} is the $\KK$-linear
  subspace of $\KK(x)$ defined by $\Exc = H(\im M)$.
\end{definition}

\begin{lemma}\label{lem:exc}
  For any~$R\in \KK(x)$,
  $R \in \im M$ if and only if $H(R) \in \Exc$.
\end{lemma}

\begin{proof}
  The direct implication is the definition of~$\Exc$.
  For the converse, assume~$H(R) = H(M(U))$ for some~$U$.
  As $(R-M(U)) - H(R-M(U)) = M(V)$ for some~$V$ by Prop.~\ref{prop:hermite-reduced}~\ref{item:H:equiv-mod-L},
   $R = M(U+V)$.
\end{proof}

The generalized Hermite reduction is not a canonical form, but it is strong
enough to ensure that~$\Exc$ is finite-dimensional over~$\KK$.

\begin{proposition}\label{prop:generators-exc}
  Over\/ $\calg$, the vector space\/ $\Exc$ is generated by the finite family
  \begin{enumerate}[(a)]
  \item $H(M (  (x-\alpha)^{-s} ))$ with $\alpha \in \Sing(M)$,
    $s > 0$ and $\ind_\alpha(-s) = 0$,
  \item $H(M (  (x-\alpha)^{-s} ))$ with $\alpha \in \Sing(M)$,
    $0 < s \leq \sigma_\alpha$,
  \item $H(M(x^s))$ with~$s \geq 0$ and $\ind_\infty(-s) = 0$,
  \end{enumerate}
  where\/~$\Sing(M) \subset \calg$ is the set of singularities of~$M$ (the zeroes
  of its leading coefficient).
\end{proposition}

\begin{proof}
  The elements $(x-\alpha)^{-s}$ ($\alpha\in\calg$, $s > 0$) and~$x^s$ ($s\geq
  0$) form a basis of~$\calg(x)$. In particular, $\Exc$, by definition, is
  generated by the $H((x-\alpha)^{-s})$ and~$H(x^s)$. By
  Proposition~\ref{prop:hermite-reduced}~\ref{item:H:reduceTa}
  and~\ref{item:H:reduceinf}, $H( M((x-\alpha)^{-s}) ) = 0$ when
  $\ind_\alpha(-s) \neq 0$ and $s < \sigma_\alpha$. Similarly, $H(M(x^s)) = 0$
  when~$\ind_\infty(-s) \neq 0$. Moreover, any~$\alpha\in \calg$ such
  that~$\ind_\alpha$ has a negative root or~$\sigma_\alpha > 0$ is a
  singularity of~$M$. Therefore, the only nonzero generators of~$\Exc$ belong
  to the set given in the statement.
\end{proof}

\begin{example}
  Let~$M = x^{10} \partial_x$. We compute $\ind_\alpha(s) =
  - \alpha^{10} s$
  for
  any~$\alpha \in \calg$ and~$\ind_\infty(s) = s$.
  Moreover~$\sigma_\alpha = -1$ for~$\alpha \not\in \{ 0, \infty\}$, $\sigma_0 = 9$
  and~$\sigma_\infty = -9$.
  It follows that
  \[ \Exc = \operatorname{Vect} \left\{ H(M(x^{-9})), \dotsc, H(M(x^{-1})) \right\} = \operatorname{Vect} \left\{ 1, x, \dotsc, x^8 \right\}. \]
\end{example}

\begin{lemma}\label{lem:modulo-finsubspace}
  Given a finite-dimensional $\KK$-linear subspace~$W\subset \KK(x)$,
  there is a unique idempotent linear map~$\rho_W : \KK(x) \to \KK(x)$
  such that:
  \begin{enumerate*}[(i)]
  \item $W = \ker \rho_W$;
  \item for any~$R\in \KK(x)$, the degree of the numerator of~$\rho_W(R)$ is
    minimal among all~$S\in\KK(x)$ with~$R - S \in W$.
  \end{enumerate*}
\end{lemma}
\noindent The following proof gives an algorithm for computing~$\rho_W$.
\begin{proof}
  When~$W \subset \KK[x]$, the value~$\rho_W(R)$ is the result of Gaussian
  elimination applied in the monomial basis to the polynomial part of~$R$ with the elements of~$W$. %

  In the general case, we write~$W = Q^{-1} V$, for some subspace~$V\subset \KK[x]$
  and~$Q\in\KK[x]$,
  and define~$\rho_W(R) = Q^{-1} \rho_V( QR )$.
  The two properties are easily checked.
\end{proof}

\begin{definition}
  The \emph{generalized Hermite reduction with respect to $M$} is the map $[\ ] :\KK(x)\to \KK(x)$
  defined by $[R] = \rho_{\,\Exc} ( H(R) )$.
\end{definition}

\begin{theorem}\label{thm:canonical}
  The map\/~$[\ ]$ is a canonical form with respect to $M$.
\end{theorem}

\begin{proof}
  We check the properties of Definition~\ref{def:canonical}.
  Let~$R\in\KK(x)$. First, $[M(R)] = 0$ because $H(M(R))\in \Exc$
  (Lemma~\ref{lem:exc}) and then $\rho_{\Exc}(H(M(R))) = 0$,
  by Proposition~\ref{prop:hermite-reduced}~\ref{item:H:equiv-mod-L}
  and the construction of~$\rho_{\Exc}$.
  Second, $R - [R] \in \im M$ because~$R - H(R) \in \im M$
  (Proposition~\ref{prop:hermite-reduced})
  and~$H(R) - \rho_{\Exc}(H(R)) \in \Exc \subset \im M$.
\end{proof}

\subsection{Rational Generalized Hermite Reduction}
\label{sec:comp-an-herm}

\begin{algo}
\begin{algorithmic}
  \Function{WHermiteRed}{$R$, $M$}
  \If{$R = 0$} \Return $0$
  \ElsIf{$R$ is a polynomial}
  \State write~$R$ as~$cx^s+ \text{(lower degree terms)}$
  \If{$\ind_\infty(-s-\sigma_\infty) \neq 0$ and~$s+\sigma_\infty \geq 0$}
  \State \Return $\textsf{WHermiteRed} \left( R - \frac{c
      M(x^{s+\sigma_\infty})}{\ind_\infty(-s-\sigma_\infty)}, M \right)$
  \Else\
   \Return $c x^s + \textsf{WHermiteRed}(R - cx^s, M)$
  \EndIf
  \Else
  \State $P \gets$ an irreducible factor of the denominator of~$R$.
  \State write~$R$ as~$\frac{A}{P^s Q}$, with~$A, Q \in \KK[x]$ and $s$ maximal.
  \If {$\operatorname{ind}_P(-s-\sigma_P) = 0$}
  \State $U \gets A/Q \mod P$.
  \State \Return $U/P^s + \textsf{WHermiteRed}(R-U/P^s, M)$
  \Else
  \State $R \gets A/Q/\ind_P(-s-\sigma_P) \mod P$
  \State \Return
  $\textsf{WHermiteRed}\left(R-M(R/P^{s+\sigma_P}),
  M\right)$
  \EndIf
  \EndIf
  \EndFunction
\end{algorithmic}
\caption[]{Rational weak Hermite reduction.
  \begin{description}
  \item[Input] $R\in \KK(x)$; $M$ a linear differential operator.
  \item[Output] The rational weak Hermite reduction of~$R$.
  \end{description}
\label{algo:HermiteRedRat}
}
\end{algo}

In most cases, computing Hermite reduction
as it is defined above would require  to work with algebraic extensions of the base field.
If $P\in \KK[x]$ is a monic irreducible polynomial and $\alpha$ a root of $P$,
the reduction can be performed simultaneously at all roots of~$P$ without
introducing algebraic extensions.

{The indicial equation} is obtained by
considering the leading coefficient of the $P$-adic expansion of $M(P^{-s})$,
see~\cite[\S4.1, p.~107]{PuSi03}.
More precisely,
there is a unique
polynomial~$\ind_P(s)$ with coefficients in~$\KK[x]/(P)$ and a unique
integer~$\sigma_P$ such that for any~$s > 0$,
\[ M\!\big( P^{-s} \big) = \ind_P(-s) P^{-s+\sigma_P} + O\!\big(
    P^{-s+\sigma_P+1} \big), \]
as $P$-adic expansions.
Since~$P$ is irreducible, $\ind_P(s)$, for a given~$s$, is either~$0$ or
invertible modulo~$P$.
For an irreducible polynomial $P\in \KK[x]$, and for~$R = U P^{-s} +
O(P^{-s+1})$, we define
\begin{equation*}
   H_P(R) =
  \begin{cases}
    UP^{-s} + H_P( R - UP^{-s} ) &\hfill \llap{\text{if $\ind_P(-\sigma_P - s) = 0$,}}  \\
    H_P\!\left( R - {M( U \ind_\alpha(-\sigma_P - s)^{-1} P^{-s-\sigma_P})} \right) & \text{otherwise,}
  \end{cases}
\end{equation*}
where~$\ind_\alpha(-\sigma_P - s)^{-1}$ is computed mod~$P$.
This is the part of our reduction which most closely resembles
the original Hermite reduction, with successive coefficients obtained by
modular inversions.

\begin{definition}
  The \emph{rational weak Hermite reduction} is the linear map~$H_\mathrm{rat} : \KK(x)\to\KK(x)$,
  defined by
\[ H_\mathrm{rat}(R) = H_\infty \big( R_{(\infty)} + \smash{\sum_{P}} H_P \big( R_{(P)}
  \big) \big),\]
where the summation runs over the irreducible factors of the denominator of~$R$
and~$R_{(P)} \in \KK[x,P^{-1}]$ denotes the polar part of the $P$-adic expansion of~$R$.
\end{definition}

The maps~$H$ and $H_\mathrm{rat}$ satisfy the same properties, \emph{mutatis mutandis}.
In particular, the latter can be used to compute a canonical form in the same
way as~$H$. Yet, both reductions are not equal (see also \S \ref{sec:absol-herm-reduct}).
For example, over~$\QQ$ with~$M = (x^2+1)\partial_x + 10x$,
 $R = (x^2+1)^{-5}$ and~$i^2+1 = 0$,
\[ H(R) = \tfrac{i}{32} \big( (x+i)^{-5} - (x-i)^{-5} \big) \quad \text{whereas}
  \quad H_\mathrm{rat}(R) = R. \]

Partial fraction decomposition and actual Hermite
reduction can be performed together. This is described in Algorithm~\ref{algo:HermiteRedRat}.
Together with the algorithm for the map~$\rho_{\Exc}$,
described in the proof of Lemma~\ref{lem:modulo-finsubspace},
we obtain an algorithm, denoted \textsf{CanonicalForm}, to compute the map~$\rho \circ H_\mathrm{rat}$ that is a
canonical form modulo~$M$.

\subsection{Variants and Improvements}

\subsubsection{Absolute Hermite reduction}\label{sec:absol-herm-reduct}
A notion of Hermite reduction that is independent from the base field
is obtained by replacing~$U \cdot P^{-s}$ with $\frac{\ud^{s-1}}{\ud x^{s-1}} \frac{U}{P}$
in the definition of~$H_P$.
Another benefit of this choice is that it is not necessary that~$P$ is irreducible to perform the reduction, but simply
that~$\ind_P(s)$ is either~$0$ or invertible. %
The denominators that appear in the computation can be factored on the fly into factors
with the required property: when some~$\ind_P(s)$ is neither~$0$
nor invertible modulo~$P$, a gcd computation gives a non-trivial divisor of~$P$.

\subsubsection{Reduction to the polynomial case}
The hypothesis that the differential operator~$M$ has polynomial coefficients
is important for the correctness of Algorithm~\ref{algo:HermiteRedRat}. To
compute canonical forms modulo an operator~$ M$ with rational coefficients, it
is sufficient to find a polynomial~$Q$ such that~$ M Q$ has polynomial
coefficients and then, to compute canonical forms modulo~$ M Q$ with the
algorithms above. Indeed, the image of~$\KK(x)$ by~$ M Q$ and~$M$ are the
same. The smallest such~$Q$ is the gcd of the denominators of the
coefficients of the adjoint of~$M$.

\subsubsection{Rational factors}
The following observation can be used to speed up the computation.
\begin{lemma}\label{rat-fact}
	Let $L, M \in \KK[x]\langle\partial_x\rangle$ and $A,B$ in~$\KK(x)$
  such that~$MA = BL$.
  If\/~$[\ ]_L$ is a canonical form \wrt~$L$, then\/
  $[\ ]_M : R\in\KK(x)\mapsto B \, [R/B]_L$
  is a canonical form \wrt~$M$.
\end{lemma}
\begin{proof}
  We check the properties of Def.~\ref{def:canonical}:
  $[ M(y) ]_M = B\, [ L(A^{-1} y) ]_L$ is 0 and
  $R - [R]_M = B \left( R/B - [R/B]_L \right)$ is in~$B (\im L) = \im M$.
\end{proof}
Lemma~\ref{rat-fact} may be used with~$A= B =\prod_\alpha(x-\alpha)^{m_\alpha}$,
where~$m_\alpha$ is the smallest negative integer root of the indicial
polynomial of~$M$ at~$\alpha$, and $0$ if none exists. This is mostly useful for equations of order~1, since the corresponding~$\alpha$ is not a singularity of the new operator, which becomes smaller.
The rational function $A$
plays the role of the \emph{shell} in previous reduction-based
algorithms~\cite{BoChChLiXi13,BoDuSa16}.

\subsection{Dimension of the Quotient with Fixed Poles}

\label{sec:bound-dimens-quot}

Let~$P \in \KK[x]$ be a squarefree polynomial
and let~$E_P = \KK\left[ x, P^{-1}
\right]$. %
Let~$\ker M \subset \KK(x)$ be the space of rational solutions of~$M$.
Let $r$ be the order of~$M$ and~$d$ the maximal degree of its coefficients.

\begin{proposition}[{Adolphson \cite[Sec.~5, Prop.~1]{Adolphson76}}]
	\label{prop:dimensioncoker}
  \begin{align*}
    \dim_{\KK} {E_P}/{M(E_P)} &= \dim_{\KK} \left(E_P \cap \ker M\right) - \sigma_\infty - \smash{\sum_{P(\alpha)=0}} \sigma_\alpha \\
                                &\leq (\deg P + 1)\cdot r + d.
  \end{align*}
\end{proposition}

\begin{proof}[Sketch of the proof]
  Let~$Z = \left\{ \alpha\in\smash{\overline \KK} \st P(\alpha)=0 \right\}$.
  Given $\deg P+1$ positive integers~$s_\infty$ and~$s_\alpha$ ($\alpha\in Z$),
  let~$E_P(s)$ denote the subspace of all~$R\in E_P$ such that the pole order at
  $\alpha$ is at most~$s_\alpha$ for~$\alpha \in Z \cup \{\infty\}$, that is
  all elements $R\in E_P(s)$ of the form
  \begin{equation*}\textstyle
    R = \sum_{\alpha\in Z} \sum_{s=1}^{s_\alpha}
    \frac{c_{\alpha,s}}{(x-\alpha)^s} + \sum_{s=0}^{s_\infty} c_{\infty,s} x^s.
  \end{equation*}
  We choose $s_\alpha$ and~$s_\infty$ large enough so that~$\ker M\subset E_P(s)$.
  Let~$t_\alpha = s_\alpha - \sigma_\alpha$ ($\alpha \in Z \cup \{\infty\}$).
  We check~$M(E_P(s))\subseteq E_P(t)$ and that a basis of~$E_P(t)/M(E_P(s))$ induces a
  basis of~$E_P/M(E_P)$.
  The bounds~$-\sigma_\alpha \leq r$, $-\sigma_\infty \leq d$ and $\dim
  \ker M \leq r$ give the inequality.
\end{proof}

\section{Creative Telescoping}\label{sec:appl-creat-telesc}

The method of creative telescoping is an approach to the computation of definite
sums and integrals of objects characterized by linear functional equations.
The notion of linear functional equation is formalized by \emph{Ore algebras}.
In this part, we consider the Ore algebra
$\mathbb{A} = \KK(x)\langle \partial_x, \partial_1,\dots,\partial_e
\rangle$,
where~$\partial_x$ is the differentiation with respect to~$x$
and~$\partial_1,\dots,\partial_e$ are arbitrary Ore operators.
In the most typical case, $\KK = \QQ(t_1,\dotsc,t_e)$
and each~$\partial_i$ is either the differentiation with respect
to~$t_i$ or the shift
$t_i \mapsto t_i+1$.

For a given function~$f$ in a function space on which~$\mathbb{A}$
acts, the annihilating ideal of~$f$ is the left ideal~$\ann f\subseteq\mathbb{A}$ of all operators that annihilate~$f$.
For example, the annihilating ideal in~$\KK(x)\langle \partial_x \rangle$ of~$f = \sin(x)$ is generated by~$\partial_x^2 + 1$
because~$\sin''(x) = -\sin(x)$.

A left ideal $\mathcal{I}$ is \emph{D-finite} if the quotient~$\mathbb{A}/\mathcal{I}$ is
a finite-dimensional vector space over~$\KK(x)$. A function is called \emph{D-finite} if its
annihilating ideal is D-finite.
We refer to~\cite{Chyzak98,Chyzak14,ChSa98} for an introduction to Ore
algebras, creative telescoping and their applications.

Given a D-finite function~$f$,
the problem of \emph{creative telescoping} is the computation of
a generating set of the \emph{telescoping  ideal} of~$f$ \wrt~$x$,
or of its residue class in $\mathbb{A}/\ann f$.
This is by definition the left ideal~$\mathcal{T}_f \subset \KK\langle
\partial_1,\dotsc,\partial_e\rangle$ of all operators~$T$ such that~$T
+ \partial_x G \in
\ann f$ for some~$G \in \mathbb{A}$; equivalently,
\[ \mathcal{T}_f = \left( \ann f + \partial_x \mathbb{A} \right)
\cap
  \KK\langle \partial_1,\dotsc,\partial_e\rangle. \]

\begin{example}\label{example:sec-ct}
  In~\S \ref{sec:ex-ct}, we use the Ore algebra
  $\KK(x)\langle \partial_x, \partial_1, \partial_2 \rangle$,
  with $\partial_1=d/dp$ and $\partial_2=S_n$ the
  shift
  \wrt~$n$. The annihilating ideal $\mathcal{I}$
  of~$F_n(p,x)$ is generated by three operators,
 one for each
  functional equation.
  It is D-finite and the quotient $\mathbb{A}/\mathcal{I}$ has
  dimension~$2$, with
  basis~$1$ and~$\partial_x$.
  The telescoping ideal of~$F_n(p,x)$
	 (or, equivalently, of~$1 \in  \mathbb{A}/\mathcal{I}$)
  is generated by~$p^2
  \partial_p^2 +p\partial_p -(n^2+p^2)$ and~$pS_n +p\partial_p - n$.
\end{example}

\subsection{Cyclic Vector}
\label{sec:cyclic-vector}

Let~$\mathcal{I}\subseteq \mathbb{A}$ be a D-finite ideal
and let~$r$ be the dimension of~$\mathbb{A}/\mathcal{I}$ over~$\KK(x)$.
We denote~$L(\gamma)$ the multiplication of an operator~$L\in \mathbb{A}$ and a
residue class~$\gamma \in \mathbb{A}/\mathcal{I}$.

Let~$\gamma\in \mathbb{A}/\mathcal{I}$ be a \emph{cyclic vector} with
respect to~$\partial_x$.
This means that~$\Gamma = \left\{ \gamma, \partial_x(\gamma), \dotsc,
  \partial_x^{r-1}(\gamma) \right\}$ is a basis
of~$\mathbb{A}/\mathcal{I}$;
or, equivalently, that every~$f\in \mathbb{A}/\mathcal{I}$ can be
written~$A_f(\gamma)$ for some~$A_f \in \KK(x)\langle \partial_x\rangle$.

Let~$L \in \KK[x]\langle\partial_x\rangle$ be a minimal
annihilating operator of~$\gamma$,
that is~$L(\gamma) = 0$ and~$L$ has order~$r$ (because~$\Gamma$ is a basis,
there is no non-zero lower order annihilating operator for~$\gamma$).
A cyclic vector always exists when~$\mathcal{I}$ is
$D$-finite~\cite{ChurchillKovacic2002,Adjamagbo88}.
It plays a role
analogous to that of primitive elements for 0-dimensional polynomial
systems.

For~$1\leq i\leq e$, we define a $\KK$-linear map~$\lambda_i : \KK(x) \to \KK(x)$
as follows.
First, we can write~$\partial_i(\gamma) = B_i(\gamma)$ for some operator $B_i \in \KK(x)\langle \partial_x \rangle$.
Next, let~$\sigma_i$ and $\delta_i$ be the maps%
\footnote{If~$\partial_i$ is the differentiation w.r.t.~$t_i$,
  then~$\sigma_i(R) = R$ and~$\delta_i(R) = {\partial R}/{\partial t_i}$.\\
  If~$\partial_i$ is the shift~$t_i \mapsto t_i + 1$, then~$\sigma_i(R) = R|_{t_i \gets
    t_i+1}$ and~$\delta_i(R) = 0$.}
such that $\partial_i R = \sigma_i(R) \partial_i +
\delta_i(R)$ for any~$R\in\KK(x)$.
Finally, we define for~$R\in\KK(x)$
\[ \lambda_i(R) = B_i^*(\sigma_i(R)) + \delta_i(R), \]
where~$B_i^*(\sigma_i(R)) \in \KK(x)$ is the result of applying the adjoint
operator $B_i^*$ to~$\sigma_i(R)$, not the operator $B_i^* \sigma_i(R)$.

\begin{proposition}\label{prop:novector}
  With the notation above:
  \begin{enumerate}[(i)]
  \item $f = A_f^*(1) \gamma + \partial_x(Q)$, for some~$Q\in \mathbb{A}/\mathcal{I}$.
  \end{enumerate}
  Moreover, for any~$R\in\KK(x)$:
  \begin{enumerate}[(i),resume]
  \item $\partial_i(R \gamma) = \lambda_i(R) \gamma + \partial_x(Q)$, for some~$Q\in
    \mathbb{A}/\mathcal{I}$.
  \item $R \gamma \in \partial_x \left( \mathbb{A}/\mathcal{I} \right)$ if and
    only if~$R \in L^*\!\left( \KK(x) \right)$.
  \end{enumerate}
\end{proposition}

\begin{proof}
  Using that~$f = A_f(\gamma)$, Lagrange's identity~\eqref{eq:Lagrange} shows that $1 A_f
  (\gamma) -
  A_f^*(1) \gamma = \partial_x(Q)$ for some~$Q$. This gives~\emph{(i)}.
Similarly, using the commutation rule
  for~$\partial_i$ and the definition
  of~$B_i$ yields~\emph{(ii)}. Property~\emph{(iii)} is shown by Abramov and
  van~Hoeij~\cite[Prop. 3]{AbHo99}.
\end{proof}

\begin{example}[Continuing Example~\ref{example:sec-ct}]
  The element $1\in \mathbb{A}/\mathcal{I}$ is a cyclic vector since~$\left\{
    1,\partial_x \right\}$ is a basis of the quotient.
\end{example}

Actual computations are performed using a Gröbner basis of~$\mathcal{I}$
and linear algebra in the finite-dimensional $\KK(x)$-vector space~$\mathbb{A}/\mathcal{I}$.

\subsection{Creative Telescoping by Reduction}
\begin{algo}
  \begin{algorithmic}
    \Function{CreativeTelescoping}{$\mathcal{I}$, $f$}
    \State $\gamma \gets$ a cyclic vector of~$\mathbb{A}/\mathcal{I}$ with respect to $\partial_x$
    \State $L \gets$ the minimal operator annihilating~$\gamma$
    \State $\lambda_1,\dotsc,\lambda_e \gets$ maps as in
    Prop.~\ref{prop:novector}
    \State $F_1 \gets \mathsf{CanonicalForm}(A_f^*(1), L^*)$
    \State $\mathcal{L} \gets [1]$ \Comment{list of monomials in
    $\partial_1,\dots,\partial_e$}
    \State $G \gets \{\}$ \Comment{Gröbner basis being computed}
    \State $Q \gets \{\}$ \Comment{Generators of the quotient}
    \State $R \gets \{\}$ \Comment{Set of reducible monomials}
     \While{$\mathcal{L}\neq\emptyset$}
       \State Remove the first element $\mu$ of $\mathcal L$
       \If{$\mu$ is a not multiple of an element of~$R$}
       \If{$\mu\neq1$}
       	   \State Pick~$i$ such that $\mu/\partial_i\in Q$
	       \State $F_{\mu} \gets \mathsf{CanonicalForm}(\lambda_i
	       (F_{\mu/\partial_i}), L^*)$
       \EndIf
       \If{$\exists$ a $\KK$-linear rel. between $F_\mu$ and $\{F_\nu\mid
           \nu\in Q\}$}
         \State $(a_\nu)_{\nu\in Q} \gets$ coeff. of the relation
         $F_\mu = \sum_{\nu\in Q} a_\nu F_\nu$
         \State Add $\mu - \sum_{\nu\in Q} a_\nu \nu$ to $G$;
         Add $\mu$ to $R$
      \Else
        \State Add $\mu$ to $Q$
        \For{$1 \leq i \leq e$}
        Append the monomial $\partial_i\mu$ to $\mathcal{L}$
        \EndFor
      \EndIf
       \EndIf
      \EndWhile
      \smallskip
    \Return $G$ %
    \EndFunction
  \end{algorithmic}
  \caption[]{Reduction-based creative telescoping algorithm
    \begin{description}
    \item[Input] $\mathcal{I}$ a D-finite ideal of~$\mathbb{A}$ and~$f \in \mathbb{A}/\mathcal{I}$
    \item[Output] Generators %
      of the telescoping ideal~$\mathcal{T}_f$
    \end{description}}
  \label{algo:ct1}
\end{algo}

We now present our algorithm (Algorithm~\ref{algo:ct1}) based on generalized Hermite reduction for
the computation of the telescoping ideal $\mathcal{T}_f$ for an element~$f$ of
some D-finite quotient~$\mathbb{A}/\mathcal{I}$.
The element~$f$ is often~$1$, as in Example~\ref{example:sec-ct}.

In the same way as Chyzak's algorithm~\cite{Chyzak00},
ours iterates over monomials in~$\partial_1,\dotsc,\partial_e$
by a strategy reminiscent of the FGLM algorithm
\cite{FaugereGianniLazardMora-1993-ECZ}.
Each iteration finds
either a new generator of
$\KK\langle\partial_1,\dotsc,\partial_e\rangle/\mathcal{T}_f$
or a new element in~$\mathcal{T}_f$.
Let $[\ ]$ be the generalized Hermite reduction with respect to~$L^*$,
the adjoint of the minimal annihilating operator of the cyclic vector~$\gamma$.
Since every visited monomial~$\mu$ (but the first)
can be written~$\partial_i \nu$
for a previously visited monomial~$\nu$,
we define~$F$ inductively by the formula $F_\mu = [\lambda_i(F_\nu)]$
and the base case $F_1 = A_f^*(f)$.
With Prop.~\ref{prop:novector} and Theorem~\ref{thm:canonical},
we check that
$\mu(f) = F_\mu \gamma + \partial_x(Q_\mu)$,
for some~$Q_\mu \in \mathbb{A}/\mathcal{I}$ and that
\begin{equation}\label{eq:3}
a_1 \mu_1 + \dotsb + a_s \mu_s \in \mathcal{T}_f \quad \Leftrightarrow\quad
  a_1 F_{\mu_1} + \dotsb + a_s F_{\mu_s}=0.
\end{equation}

\begin{theorem}
  On input~$\mathcal{I}$, Algorithm~\ref{algo:ct1} terminates if and only
  if the telescoping ideal~$\mathcal{T}_f$ is D-finite.
  It outputs a Gröbner basis of~$\mathcal{T}_f$ for the grevlex monomial ordering.
\end{theorem}

\begin{proof}
  By construction, when a monomial is added to the set~$R$, it is not a
  multiple of another monomial in~$R$. By Dickson's lemma~\cite{Dickson13},
  this may happen only finitely many times.

  The way~$\mathcal{L}$ is filled ensures that when a monomial~$\mu$ is visited,
  every smaller monomial has been visited or is a multiple of a reducible
  monomial.
  This implies, by induction, that~$Q$ is the set of all non-reducible
  monomials that are smaller than~$\mu$, when~$\mu$ is visited.

  If~$\mathcal{T}_f$ is not D-finite, then there are infinitely many non-reducible
  monomials and the algorithm does not terminate.
  Otherwise, the algorithm terminates, since neither
  $Q$ nor~$R$ may
  grow indefinitely.

  To check that~$G$ is a Gröbner basis, we note that:
  $G \subset \mathcal{T}_f$, by the equivalence~\eqref{eq:3};
  the leading monomials of the elements of~$G$ are the elements of~$R$;
  and every leading monomial of an element~$\mathcal{T}_f$ (that is a reducible monomial) is a multiple of an element of~$R$.
\end{proof}

\subsection{Variants and Improvements}
\subsubsection{Different term order} As stated, Algorithm~\ref{algo:ct1} computes
relations by increasing total degree in~$(\partial_1,\dots,\partial_e)$. To choose a different term order, it is sufficient to change the selection of the monomial~$\mu$ at the beginning of the loop and select the smallest one for the given order instead.

\subsubsection{Different termination rule} Instead of waiting for the
list~$\mathcal{L}$ to be empty, one can stop as soon as a
relation is found, and then it is the minimal one for the chosen
term order. This variant does not require a D-finite
ideal to terminate. Another possibility is to stop as soon as
the degree of~$\mu$ is larger than a predefined bound, returning all
the relations that exist below this bound.
\subsubsection{Certificates} While an important point of the reduction-based
approaches to creative telescoping is to avoid the computation of
certificates (in contrast with Chyzak's and Koutschan's algorithms
that require their computation), it is also possible to modify the
algorithm so that it returns a certificate for each element of the basis. Indeed, a certificate of the generalized Hermite reduction of \S\ref{sec:canonical-forms} can be propagated through the algorithms.

\subsection{D-finiteness of the Telescoping Ideal}

In the general case, the telescoping ideal~$\mathcal{T}_f$ of a D-finite
function~$f$ need not be D-finite. %
However, when the auxiliary operators~$\partial_1,\dotsc,\partial_e$ are
differentiation operators (as opposed to shift operators for example),
then~$\mathcal{T}_f$ is always D-finite if~$f$ is;
this is a well-known result in the theory of D-finiteness and
holonomy~\cite{Takayama1992}.
We give here a new proof of this fact, as a corollary of a more general
sufficient condition for general Ore operators.

\begin{definition}
  A D-finite function $f$ is \emph{singular} (\wrt $\partial_x$) at~$\alpha\in\calg$  if
  every nonzero operator~$L \in \KK(x)\langle \partial_x \rangle$ such
  that~$L(f) = 0$ is singular at~$\alpha$.
  The \emph{singular set} (\wrt $\partial_x$) of~$f$, denoted~$\Sing(f)$, is the set of all singular points of~$f$.
\end{definition}

Let~$\Theta = \{1, \partial_1,\partial_2,\dotsc,\partial_1^2,
\partial_1\partial_2,\dotsc\}$ be the set of all monomials in the
variables~$\partial_1,\dotsc,\partial_e$.

\begin{theorem}\label{thm:dfinite-telescoping}
  For any D-finite function~$f$, if\/ $\bigcup_{\mu\in\Theta} \Sing(\mu(f))$
  is finite, then~$\mathcal{T}_f$ is D-finite.
\end{theorem}

\begin{proof}
  Let~$\gamma$ be a cyclic vector of~$\mathbb{A}/\ann f$ \wrt $\partial_x$
  with minimal annihilating operator $L\in\KK(x)\langle \partial_x \rangle$ of order~$r$.
  For~$\mu\in\Theta$, let~$A_{\mu(f)} \in \KK(x) \langle \partial_x \rangle$
  of order $< r$ be such that~$\mu(f) = A_{\mu(f)}(\gamma)$, as
  in~\S\ref{sec:cyclic-vector},
  and let~$R_\mu = A_{\mu(f)}^*(1) \in \KK(x)$, so that
  $\mu(f) = R_\mu \gamma + \partial_x(G_\mu)$, for some~$G_\mu\in
  \mathbb{A}/\ann f$.
  By Proposition~\ref{prop:novector},
  the $\KK$-linear map $\phi : \KK\langle{}\partial_1,\dotsc,\partial_n\rangle \to
  \KK(x)$ defined by~$\phi(\mu) = R_\mu$
  induces an injective map
  $\KK\langle{}\partial_1,\dotsc,\partial_n\rangle/\mathcal{T}_f \to
    \KK(x)/\im L^*$.
  The telescoping ideal $\mathcal{T}_f$ is D-finite if and only if the image
  of this map is finite-dimensional.
  In view of
  Proposition~\ref{prop:dimensioncoker}, it suffices to show that
  the poles of all the~$R_\mu$ lie in a finite subset of~$\calg$.
  This is obtained by proving that
  \begin{equation}\label{eq:4}
   \bigcup_{\mu\in\Theta} \operatorname{poles}(R_\mu) \subseteq \Sing(L) \cup
    \bigcup_{\mu\in\Theta} \Sing(\mu(f)),
  \end{equation}
  where $\Sing(L)$ is the set of zeroes of the leading coefficient of~$L$.

  Indeed, let~$\mu\in\Theta$ and~$\alpha\in\calg$ that is not in the right-hand side.
  We now prove that no coefficient of~$A_{\mu(f)}$ has a pole
  at~$\alpha$, from where it follows that neither has~$R_\mu = A_{\mu(f)}^*(1)$.
  By the hypothesis on~$\alpha$, there exists~$M\in\KK(x)\langle
  \partial_x \rangle$ an annihilating
  operator of~$\mu(f)$ regular at~$\alpha$. It satisfies
  $M A_{\mu(f)}(\gamma) = M(\mu(f))=0$ and by minimality of~$L$ it
  follows that $M A_{\mu(f)} = BL$ for some operator~$B$. As a
  consequence, $0,1,\dots,r-1$ are roots of the indicial
  polynomial of~$MA_{\mu(f)}$. Write~$A_{\mu(f)} = \sum_{i=0}^{r-1}
  a_i \partial_x^i$,
  for some~$a_i\in\KK(x)$ and
  let~$j$ be the maximal index with~$\ord_\alpha a_j =
  \min_i \ord_\alpha a_i$. Then~$j\in\{0,\dots,r-1\}$ and~$\ord_\alpha
  A_{\mu(f)}\!\left((x-\alpha)^j\right)={\min_i \ord_\alpha a_i}$. Then this last quantity
  is a zero of the indicial polynomial of~$M$, which implies that it
  is nonnegative and thus that none of the~$a_i$ has
  a pole
  at~$\alpha$.
\end{proof}

Recall that for~$1 \leq i \leq e$,
the Ore operator~$\partial_i$ satisfies
a commutation relation $\partial_i a = \sigma_i(a) \partial_i + \delta_i(a)$
for any~$a\in \KK$,
where $\sigma_i$~is an endomorphism of~$\KK$
and $\delta_i$~is a $\sigma_i$-derivation of it.
When $\partial_i$~is a differentiation operator, $\sigma_i = \operatorname{id}_{\KK}$.

\begin{corollary}
  If~$\partial_1,\dotsc,\partial_e$ are differentiation operators, then
  $\mathcal{T}_f$ is D-finite for any D-finite function~$f$.
\end{corollary}

\begin{proof}
  It is sufficient to  check that~$\Sing(\mu(f))\subset
  \Sing(f)$ for any monomial~$\mu\in\Theta$ and then conclude by
  Theorem~\ref{thm:dfinite-telescoping}.

  Let $M\in\KK[x]\langle \partial_x \rangle$ be an annihilating operator of $g=\mu(f)$ regular
  at~$\alpha\in\calg\setminus\Sing(f)$. The commutation rules for the differential operators
  imply that~$\partial_i M = M\partial_i + R$, for some $R \in\KK[x]\langle
  \partial_x\rangle$. In particular, we obtain the inhomogeneous differential
  equation~$M(\partial_i(g)) = R(g)$ for~$\partial_i(g)$. Since~$\alpha$ is neither a
  singularity of~$M$ nor of~$R(g)$, it follows that it is not a singularity
  of~$\partial_i(g)$.
\end{proof}

For the case of general Ore operators, we obtain with a similar proof the
following result.

\begin{corollary}\label{coro:finite-S}
  For any D-finite function~$f$, if there is a finite set $S\subset\calg$ such that:
(i) $\sigma_i(S)\subseteq S$ for any~$1\leq i\leq e$ and (ii) $\Sing(f) \subseteq S$,
  then~$\mathcal{T}_f$ is D-finite.
\end{corollary}

\section{Experiments}

\begin{table}[t]
  \centering
  \begin{tabular}{lccccccc}
    \toprule
    Integral& \eqref{eq:9} & \eqref{eq:10} & \eqref{eq:8} &\eqref{eq:new}& \eqref{eq:7} & \eqref{eq:int1} & \eqref{eq:int2}  \\ \midrule
    \textsf{redct} &  13 s        & > 1h & > 1h &1.5 s& 1.5 s         &    165 s &     53 s \\
    \textsf{HF-CT} &  19 s        &    253 s &     45 s &232 s& 516 s         & >1h & >1h \\
    \textsf{HF-FCT}& 1.9 s* &    2.3 s&    5.3 s &>1h& 2.3 s* & 5.4 s & 2.2 s* \\\bottomrule
  \end{tabular}
  \medskip
  \caption{
    Comparative timings on several instances of creative telescoping.
    Rows are \textsf{redct} (new algorithm);
    Koutschan's \textsf{HolonomicFunctions}, using functions \texttt{Annihilator}
    and \texttt{CreativeTelescoping} (HF-CT); idem, using \texttt{FindCreativeTelescoping} (HF-FCT), a heuristic that does not necessarily find the minimal operators (indicated by *).
    All examples were run on the same machine, with the latest versions of Maple and Mathematica. }
  \label{tab:timings}
\end{table}

We present the results of a preliminary Maple implementation called
\textsf{redct}\footnote{
Available with example sessions at
\url{https://specfun.inria.fr/chyzak/redct/}.}.
Comparison is done with Koutschan's \textsf{HolonomicFunctions}
package~\cite{Koutschan2009},
the best available code for creative telescoping.
Timings are given in Table~\ref{tab:timings}\footnote{
When our code does not terminate, time is spent computing the
exceptional set. This seems to be due to
apparent singularities of the operators, that become true
singularities of their adjoint. Ways of circumventing this issue are
under study.}.

\paragraph{Koutschan's examples}
\label{sec:CKexs}

Koutschan's example session~\cite{Koutschan2017} contains~40 integrals on which we tested our code.
In most cases, our code compares well with \textsf{HolonomicFunctions}.
There are 37 easy cases,
all of whose telescopers are found in 3.5~sec. by \textsf{redct}, while 16~sec. are needed by \textsf{HolonomicFunctions} (but that includes certificates).
The three other examples are (the nature of the parameters is indicated in
the brackets, $\smash{C_n^{(\alpha)}}$ denotes Gegenbauer polynomials,
and $J_1$, $I_1$,
etc. Bessel functions):
\begingroup\eqsmall
\begin{gather}\label{eq:9}
\int{\frac{2J_{m+n}(2tx)T_{m-n}(x)}{\sqrt{1-x^2}}\,\ud x} \quad \text{[diff. $t$, shift $n$ and~$m$]}, \\
\label{eq:10}
  \int_0^1 C_n^{(\lambda)}(x) C_m^{(\lambda)}(x) C_{\ell}^{(\lambda)}(x)
  (1-x^2)^{\lambda - \frac12} \,\ud x  \quad \text{[shift $n$, $m$, $\ell$]}, \\
\label{eq:8}
  \int_0^\infty x J_1(ax) I_1(ax) Y_0(x) K_0(x) \,\ud x \quad \text{[diff. $a$]}.
\end{gather}
\endgroup

\paragraph{Longer examples}
We mention a few examples, some involving
Gegenbauer polynomials~\cite[2.21.18.2, 2.21.18.4]{PrBrMa88}, that take more time. The advantage of a reduction-based approach becomes visible.
\begin{small}
\begin{gather}
  \int{\tfrac{n^2+x+1}{n^2+1}\left(\tfrac{(x+1)^2}{(x-4)(x-3)^2(x^2-5)^3}\right)^{\!n}\!\sqrt{x^2-5}\,e^{\frac{x^3+1}{x(x-3)(x-4)^2}}\ud x}\;\,\text{[shift $n$]}, \label{eq:new} \\
  \int{C_m^{(\mu)}(x)C_n^{(\nu)}(x)(1-x^2)^{\nu-1/2}\,\ud x} \quad \text{[shift $n$, $m$, $\mu$, $\nu$]}, \label{eq:7}
\end{gather}
\begin{gather}
  \int{x^\ell C_m^{(\mu)}(x)C_n^{(\nu)}(x)(1-x^2)^{\nu-1/2}\,\ud x} \quad \text{[shift $\ell$, $m$, $n$, $\mu$, $\nu$]}, \label{eq:int1} \\
\begin{split}
\int{(x+a)^{\gamma+\lambda-1}(a-x)^{\beta-1}C_m^{(\gamma)}(x/a)
  C_n^{(\lambda)}(x/a)\,\ud x},\\
\quad \text{[diff. $a$, shift $n,m,\beta,\gamma,\lambda$]}.\label{eq:int2}
\end{split}
\end{gather}
\end{small}

\section{Conclusion}
A closer look at our algorithm reveals several aspects of the
complexity of creative telescoping. To simplify the discussion, we
restrict to the bivariate case and
measure the \emph{arithmetic complexity}, obtained by counting
arithmetic operations in~$\QQ$. We look for
bounds in terms of the
\emph{input size} (order and degree of the operators at hand).

In this setting, \emph{the complexity of computing~$\mathcal{T}_f$ is
not
bounded polynomially} (whatever the algorithm). Consider for instance,
the integral representation of Hermite polynomials
\[ H_n(t) = \frac{2^n}{i \sqrt{\pi}} \int_{-i\infty}^{i\infty} (t+x)^n e^{x^2}
  \ud x. \]
If one computes a telescoper over~$\QQ(n,t)$, then our algorithm
produces the classical  differential equation $y''+2n y = 2t y'$.
However, if~$n$ is a given positive integer then the
minimal telescoper is the first-order factor~$H_n(t)\partial_t-H_n'
(t)$, with
coefficients of degree~$n$.
Its size
\emph{is exponential in the bit
size} of the input.
Thus, \emph{no algorithm computing the minimal telescoper can
run in polynomial complexity}.

However, in the frequent cases like this one where the set~$S$ of
singularities discussed in Corollary~\ref{coro:finite-S} is bounded
polynomially in terms of the size of the input, then the
dimension of the quotient and therefore \emph{the order of the
telescopers
is bounded polynomially} as a consequence of Adolphson's result
(Proposition~\ref{prop:dimensioncoker}).
The non-polynomial cost of minimality thus resides only in the degree
of the coefficients. %
Note that in the differential case,
polynomial time computation of non-minimal telescopers is also achieved by well-known methods in
holonomy theory, e.g.,~\cite[proof of Lemma~3]{Lipshitz88}.

In our algorithm, the non-polynomial complexity arises first in the computation
of the exceptional set~$\Exc$ and next in the reductions by the
elements of this set. Removing this part of the computation and
using the weak Hermite reduction yields a weak form of the
algorithm that does not find minimal telescopers
but runs in polynomial complexity, if the set~$S$ has
polynomial size.

\medskip \paragraph{Acknowledgement.}
This work was supported in part by
FastRelax ANR-14-CE25-0018-01.

\bibliographystyle{abbrv}

\end{document}